\documentclass{article}
\usepackage{amsmath,amssymb,amsthm}
\usepackage{tikz}
\usepackage{mathtools}
\usepackage[ruled,vlined,linesnumbered]{algorithm2e}
\SetKwInput{KwInput}{Input}
\SetKwInput{KwOutput}{Output} 
\SetKwInput{KwInitiation}{Initiation}
\SetKwInput{KwRemove}{Remove}
\SetKwInput{KwAppend}{Append}
\SetKwInput{KwReturn}{Return}

\newtheorem{thm}{Theorem}[section]
\newtheorem{cor}[thm]{Corollary}
\newtheorem{lem}[thm]{Lemma}
\newtheorem{prop}[thm]{Proposition}
{\theoremstyle{definition}

\newtheorem{exa}[thm]{Example}

\newtheorem{rem}[thm]{Remark}
}

\usepackage{url,graphicx}
\usepackage{caption}
\usepackage{subcaption}
\usepackage{todonotes}
\usepackage{tikz,pgf}
\usetikzlibrary{calc}

\newcommand{\R}{\mathbb R}

\newcommand{\zero}{\vec{0}}

\title{Representing Piecewise Linear Functions by Functions with Small Arity\thanks{The research reported in this paper has been partly funded by BMK, BMDW, and the Province of Upper Austria in the frame of the COMET Programme managed by FFG in the COMET Module S3AI.}}

\author{Christoph Koutschan, RICAM \\
        Bernhard Moser, SCCH \\
	Anton Ponomarchuk, RICAM \\
	Josef Schicho, RISC JKU} 

\begin{document}

\maketitle

\begin{abstract}
    A piecewise linear function can be described in different forms: as an arbitrarily nested expression of $\min$- and $\max$-functions, as a difference of two convex piecewise linear functions, or as a linear combination of maxima of affine-linear functions. 
    In this paper, we provide two main results: first, we show that for every piecewise linear function there exists a linear combination of $\max$-functions with at most $n+1$ arguments, and give an algorithm for its computation. Moreover, these arguments are contained in the finite set of affine-linear functions that coincide with the given function in some open set. Second, we prove that the piecewise linear function $\max(0, x_{1}, \ldots, x_{n})$ cannot be represented as a linear combination of maxima of less than $n+1$ affine-linear arguments. This was conjectured by Wang and Sun in 2005 in a paper on representations of piecewise linear functions as linear combination of maxima.
\end{abstract}

\section{Introduction}

The mathematical model of a neural network is a directed graph without cycles~\cite{Goodfellow:16}, where each vertex with in-degree~$0$ stands for an input parameter ranging over~$\R$, and vertices with out-degree~$0$ stand for the output values. Each vertex with positive in-degree is called a neuron. Each neuron has finitely
many input values corresponding to the incoming edges. The neuron applies an affine-linear function to the vector of these input values, followed by a non-linear activation function. We assume that the activation function is the function $x\mapsto\max(x,0)$, also known as the ReLU function (Rectified Linear Unit, see \cite{Hinton_Nair:10}). 
The output of one neuron may be the input for another neuron, which is indicated by a directed edge between the two vertices in the graph. Figure~\ref{fig:nn} shows a ReLU network that computes a bivariate piecewise linear function.

\begin{figure}[h]
\begin{center}
\begin{tikzpicture}[scale=1.5]
  \node[circle, minimum size=25, draw] (n11) at (0,3) {$x_1$};
  \node[circle, minimum size=25, draw] (n12) at (0,1) {$x_2$};
  \node[circle, minimum size=25, draw] (n21) at (3,4) {$2$};
  \node[circle, minimum size=25, draw] (n22) at (3,2) {$-3$};
  \node[circle, minimum size=25, draw] (n23) at (3,0) {$1$};
  \node[circle, minimum size=25, draw] (n31) at (6,2) {$8$};
  \draw[->,thick] (n11) to node[above] {$-1$} (n21);
  \draw[->,thick] (n11) to node[pos=0.4,above] {$2$} (n22);
  \draw[->,thick] (n12) to node[pos=0.3,below] {$3$} (n21);
  \draw[->,thick] (n12) to node[above] {$5$} (n23);
  \draw[->,thick] (n21) to node[above] {$4$} (n31);
  \draw[->,thick] (n22) to node[above] {$-5$} (n31);
  \draw[->,thick] (n23) to node[above] {$6$} (n31);
\end{tikzpicture}
\caption{A neural network computing the piecewise linear function
$(x_1,x_2)\mapsto 4\max(-x_1+3x_2+2,0)-5\max(2x_1-3,0)+6\max(5x_2+1,0)+8$.
The activation function for each neuron is the ReLU function $x\mapsto\max(0,x)$.
The network has depth 2.
}
\label{fig:nn}
\end{center}
\end{figure}
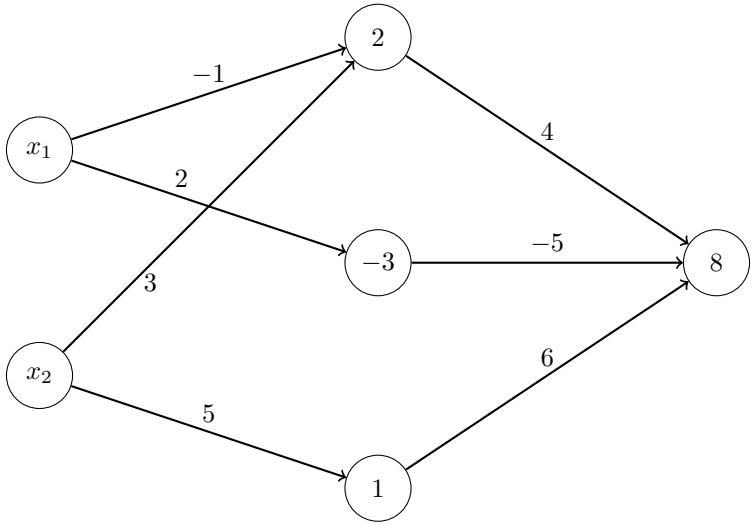

Algebraically, the depth of a ReLU network corresponds to the depth of nestings of ReLU functions in the expression
determined by the network. The function $\max(x_1,\dots,x_n)$ can be written as a composition of binary $\max$-function with nesting depth
equal to $\lceil\log_2(n)\rceil$, and every binary $\max$ can be written in terms of a ReLU function by the identity
$\max(x,y)=x+\max(y-x,0)$. This shows that any piecewise linear function that can be written as a linear combination
of maxima of at most $n$ affine-linear functions can be realized by a ReLU network of depth at most $\lceil\log_2(n)\rceil$.
This was already observed in \cite{Wang_Sun:05}.

It is well-known~\cite{Kripfganz_Schulze:87} that every piecewise linear function is the difference of two
convex piecewise linear functions; see also \cite{Schlueter_Darup:21} for a more efficient decomposition algorithm.
It follows that every piecewise linear function is a linear combination of maxima of affine-linear functions, see~\cite{Chua:88, Tarela_Martines:99, Ofchinnikov:02, Arora:18}.
The paper \cite{Wang_Sun:05} addresses the problem of minimizing the largest number of arguments of the maxima appearing
in such a linear combination. 
The authors show that every piecewise linear function defined on $\R^n$ can be written
as a linear combination of maxima of at most $n+1$ affine-linear functions. The authors also express their
conviction that this bound is optimal. In their own words, ``it seems impossible'' that the function $\max(0,x_1,\dots,x_n)$
has an expression as a linear combination of maxima of at most $n$ affine-linear functions.

This paper is structured as follows. In Section~\ref{sec:height_bound} we show that for any piecewise linear function $f\colon \R^{n}\rightarrow \R$ there exists an integral linear combination of maxima of at most $n+1$ affine-linear functions. Note that in the case $n=1$ this corresponds to the fundamental theorem of tropical algebra. In Section~\ref{sec:exmp} we give an algorithm for finding such linear integral combination of maxima with examples. In Section~\ref{sec:duality} we recall some notions about Minkowski-addition, duality between the set of non-empty convex polytopes in $\R^{n}$ and the set of convex and positively homogeneous piecewise linear functions of degree~$1$. 
In Section~\ref{sec:lb} we give a proof for the conjecture in \cite{Wang_Sun:05}, thereby showing that their bound for the number of $\max$-arguments is indeed optimal. The proof is based on properties of the Minkowski-addition
of convex polytopes.

\section{An Upper Bound for the Height}
\label{sec:height_bound}
For a given piecewise linear function $f\colon\R^n\to\R$, we define the height $H(f)$ as the smallest integer~$k$ such that $f$ is 
a linear combination of maxima of at most $k+1$ affine-linear functions. In this section we give a new proof for the bound $H(f)\le n$
that was first shown in \cite{Wang_Sun:05}. We also show that the arguments for the maxima can be chosen
among the constituents of $f$. These are defined as the affine-linear functions that coincide with $f$ in some open subset of~$\R^n$.

\begin{lem} \label{lem:ST}
Let $R$ be a finite set of affine-linear functions on $\R^n$ such that $|R|\ge n+2$. Then there exists a decomposition
into two non-empty disjoint subsets, $R=S\uplus T$, such that for all points $x\in\R^n$, we have 
\[ \max_{g\in T}g(x)\ge \min_{g\in S}g(x) . \]
\end{lem}

\begin{proof}
The derivative of every function in $R$ is a linear function in $(\R^n)^\ast$.
The vector space $(\R^n)^\ast$ has dimension $n$. Therefore the set of derivatives $\{g'\mid g\in R\}$ is affinely dependent, i.e.,
there exist real numbers $\alpha_g$, $g\in R$, not all equal to zero, such that
\[
  \sum_{g\in R} \alpha_gg'=0, \quad \sum_{g\in R} \alpha_g = 0.
\] 
We set $T:=\{ g\mid \alpha_g>0\}$ and $S:=\{ g\mid \alpha_g\le 0\}$.
Without loss of generality, we may assume $\sum_{g\in T} \alpha_g=1$ and $\sum_{g\in S} \alpha_g=-1$ -- if not, we multiply all 
$\alpha_g$ by a suitable positive constant. We also set $\beta_g:=-\alpha_g$ for $g\in S$.
Then the function $h:=\sum_{g\in T}\alpha_gg-\sum_{g\in S}\beta_gg$ has derivative zero and therefore $h$ equals to a constant~$c$.
Let us assume $c\ge 0$. Then we get
\[
  \max_{g\in T}g(x) \ge
  \sum_{g\in T}\alpha_gg(x) =
  \sum_{g\in S}\beta_gg(x)+c \ge
  \sum_{g\in S}\beta_gg(x) \ge \min_{g\in S}g(x)
\]
for all $x\in\R^n$.

If $c<0$, then we redefine $T:=\{ g\mid \alpha_g<0\}$ and $S:=\{ g\mid \alpha_g\ge 0\}$ and get a similar chain of inequalities.
\end{proof}

\begin{lem} \label{lem:binomi}
Let $A$ be a finite non-empty set. Then
\[ \sum_{T\subseteq A} (-1)^{|T|} = 0 . \]
\end{lem}

\begin{proof}
\[
  \sum_{T\subseteq A} (-1)^{|T|} = \sum_{i=0}^{|A|} \binom{|A|}{i} (-1)^i = (1-1)^{|A|} = 0 .
  \qedhere
\]
\end{proof}

\begin{lem} \label{lem:id}
With $S,T$ as above, we have the equality
\[ \sum_{M\subseteq S} (-1)^{|M|} \max_{g\in M\cup T} g(x) = 0 \]
for all $x\in\R^n$.
\end{lem}

\begin{proof}
For all $x\in\R^n$ except those in a finite union of hyperplanes, the values $g(x)$, $g\in R$ are pairwise distinct. 
We only need to prove the equality for $x$ under this assumption; then it follows for the remaining places by continuity.
So, let us fix $x$ such that the values $g(x)$, $g\in R$ are pairwise distinct. Then
\[ \sum_{M\subseteq S} (-1)^{|M|} \max_{g\in M\cup T} g(x) = \sum_{u\in R}  \left(
	\sum_{\substack{ 
		M\subseteq S \\
		u\in M\cup T \\
		\forall g\in M\cup T: g(x)\le u(x)
	}} 
  (-1)^{|M|} \right) u(x) \]
We claim the inside sum is~$0$, for each $u\in R$. 
We distinguish three cases.

Case 1: $u\in T$ and $\exists g\in T: g(x)> u(x)$. 
Then there is no $M$ such that the conditions under the sum sign are fulfilled. Then the inner sum is the sum over the empty set which is $0$ by definition.

Case 2: $u\in S$. 
We define $S_u$ as the set of all $g\in S$ such that $g(x)\le u(x)$.
Note that $S_u$ is non-empty, since it contains $u$. Then
\[ \sum_{\substack{ 
                M\subseteq S \\
                u\in M\cup T \\
                \forall g\in M\cup T: g(x)\le u(x)
        }} (-1)^{|M|} = \sum_{M\subseteq S_u} (-1)^{|M|} = 0 \]
by Lemma~\ref{lem:binomi}.

Case 3: $u\in T$ and $\forall g\in T: g(x)\le u(x)$.
We define $S_u$ as the set of all $g\in S$ such that $g(x)\le u(x)$. 
The set $S_u$ is non-empty because of Lemma~\ref{lem:ST}. As in the previous case, we get
\[ \sum_{\substack{ 
                M\subseteq S \\
                u\in M\cup T \\
                \forall g\in M\cup T: g(x)\le u(x)
        }} (-1)^{|M|} = \sum_{M\subseteq S_u} (-1)^{|M|} = 0 .
    \qedhere
\]
\end{proof}

\begin{thm} \label{thm:ws}
Every piecewise linear function $f\colon\R^n\to\R$ can be written as an integral linear combination
of maxima of at most $n+1$ affine-linear functions. Moreover, the affine-linear functions can be
chosen among the constituents of~$f$.
\end{thm}

\begin{proof}
By \cite{Tarela_Martines:99,Ofchinnikov:02}, the function $f$ can be written as maximum of minima of constituents.
Using the first identity in \cite[Lemma~1]{Wang_Sun:05}, namely
\[ \max(a,\min(b,c)) = \max(a,b)+\max(a,c)-\max(a,b,c) , \]
we can rewrite this expression as a linear combination of maxima of constituents.
By Lemma~\ref{lem:id} and Lemma~\ref{lem:ST}, we may replace any maximum of more than $n+1$
constituents by a linear combination of maxima of fewer constituents.
\end{proof}

\begin{rem}
The proof in this section shows a slightly stronger statement: $f$ can be written as an integral linear combination of maxima
of constituents whose derivatives are affinely independent. We will see in Section~\ref{sec:lb} that the maximum of affine-linear
functions that are affinely independent cannot be expressed as linear combination of maxima of fewer affine-linear functions.
\end{rem}

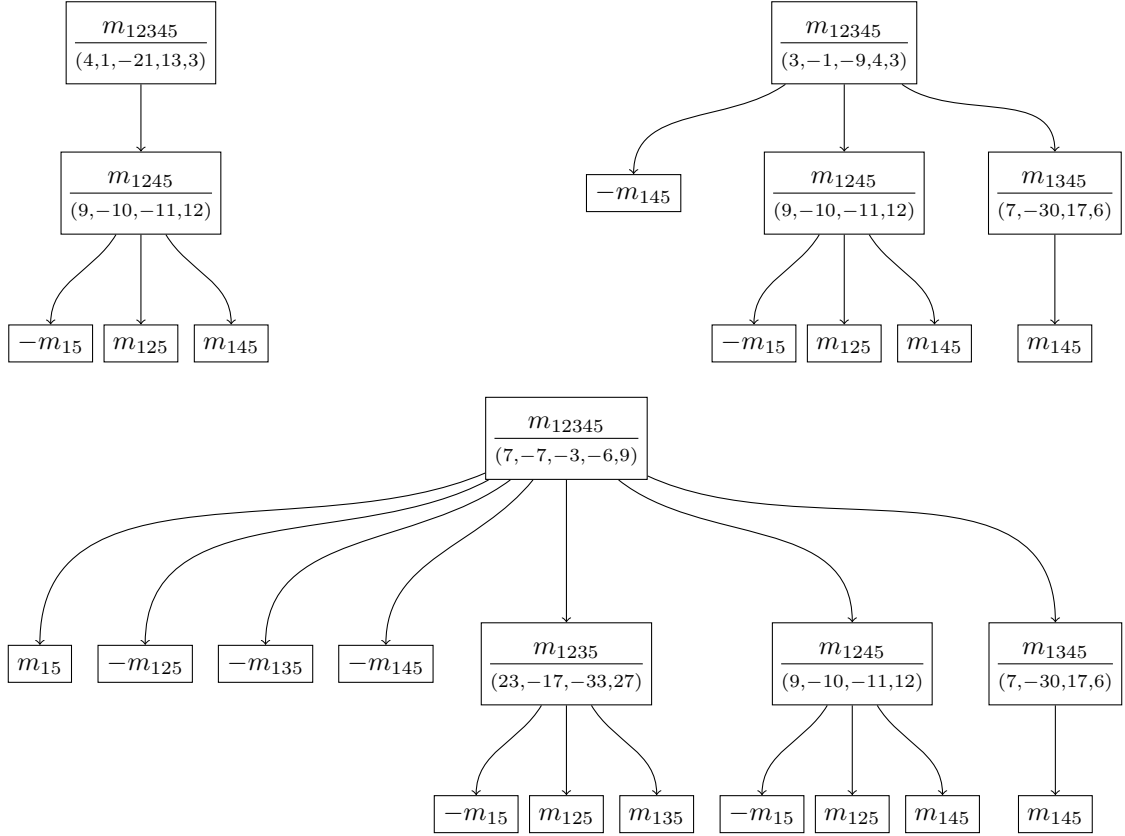
\begin{figure}
\newcommand{\mx}[3]{$\begin{array}{@{}c@{}}#1m_{#2} \\ \hline {\scriptstyle#3}\end{array}$}
\newcommand{\mxl}[2]{$\vphantom{-}{#1m_{#2}}$}
\centerline{
\begin{tikzpicture}
\node[draw] (n11) at (0,4) {\mx{}{12345}{(4,1,-21,13,3)}};
\node[draw] (n21) at (0,2) {\mx{}{1245}{(9,-10,-11,12)}};
\node[draw] (n31) at (-1.2,0) {\mxl{-}{15}};
\node[draw] (n32) at (0,0) {\mxl{}{125}};
\node[draw] (n33) at (1.2,0) {\mxl{}{145}};
\draw[->] (n11) to (n21);
\draw[->] (n21) to[out=239.0,in=90] (n31);
\draw[->] (n21) to[out=270.0,in=90] (n32);
\draw[->] (n21) to[out=301.0,in=90] (n33);
\begin{scope}[xshift=9.35cm]
\node[draw] (n11) at (0,4) {\mx{}{12345}{(3,-1,-9,4,3)}};
\node[draw] (n21) at (-2.8,2) {\mxl{-}{145}};
\node[draw] (n22) at (0,2) {\mx{}{1245}{(9,-10,-11,12)}};
\node[draw] (n23) at (2.8,2) {\mx{}{1345}{(7,-30,17,6)}};
\node[draw] (n31) at (-1.2,0) {\mxl{-}{15}};
\node[draw] (n32) at (0,0) {\mxl{}{125}};
\node[draw] (n33) at (1.2,0) {\mxl{}{145}};
\node[draw] (n34) at (2.8,0) {\mxl{}{145}};
\draw[->] (n11) to[out=215.5,in=90] (n21);
\draw[->] (n11) to[out=270.0,in=90] (n22);
\draw[->] (n11) to[out=324.5,in=90] (n23);
\draw[->] (n22) to[out=239.0,in=90] (n31);
\draw[->] (n22) to[out=270.0,in=90] (n32);
\draw[->] (n22) to[out=301.0,in=90] (n33);
\draw[->] (n23) to[out=270.0,in=90] (n34);
\end{scope}
\end{tikzpicture}
}
\bigskip

\centerline{
\begin{tikzpicture}
\node[draw] (n11) at (0,5) {\mx{}{12345}{(7,-7,-3,-6,9)}};
\node[draw] (n21) at (-7,2) {\mxl{}{15}};
\node[draw] (n22) at (-5.6,2) {\mxl{-}{125}};
\node[draw] (n23) at (-4,2) {\mxl{-}{135}};
\node[draw] (n24) at (-2.4,2) {\mxl{-}{145}};
\node[draw] (n25) at (0,2) {\mx{}{1235}{(23,-17,-33,27)}};
\node[draw] (n26) at (3.8,2) {\mx{}{1245}{(9,-10,-11,12)}};
\node[draw] (n27) at (6.5,2) {\mx{}{1345}{(7,-30,17,6)}};
\node[draw] (n31) at (-1.2,0) {\mxl{-}{15}};
\node[draw] (n32) at (0,0) {\mxl{}{125}};
\node[draw] (n33) at (1.2,0) {\mxl{}{135}};
\node[draw] (n34) at (2.6,0) {\mxl{-}{15}};
\node[draw] (n35) at (3.8,0) {\mxl{}{125}};
\node[draw] (n36) at (5,0) {\mxl{}{145}};
\node[draw] (n37) at (6.5,0) {\mxl{}{145}};
\draw[->] (n11) to[out=203.2,in=90] (n21);
\draw[->] (n11) to[out=208.2,in=90] (n22);
\draw[->] (n11) to[out=216.9,in=90] (n23);
\draw[->] (n11) to[out=231.3,in=90] (n24);
\draw[->] (n11) to[out=270.0,in=90] (n25);
\draw[->] (n11) to[out=321.7,in=90] (n26);
\draw[->] (n11) to[out=335.2,in=90] (n27);
\draw[->] (n25) to[out=239.0,in=90] (n31);
\draw[->] (n25) to[out=270.0,in=90] (n32);
\draw[->] (n25) to[out=301.0,in=90] (n33);
\draw[->] (n26) to[out=239.0,in=90] (n34);
\draw[->] (n26) to[out=270.0,in=90] (n35);
\draw[->] (n26) to[out=301.0,in=90] (n36);
\draw[->] (n27) to[out=270.0,in=90] (n37);
\end{tikzpicture}
}
\caption{Application of our Algorithm~\ref{alg:lemma_id} ``reduceMax''
to the input function $\max(g_1,g_2,g_3,g_4,g_5)$ where
$g_1(x_1,x_2)=3 x_1-4 x_2+1$,
$g_2(x_1,x_2)=-3 x_1-x_2-2$,
$g_3(x_1,x_2)=2 x_1+x_2-1$,
$g_4(x_1,x_2)=3 x_1+2 x_2+2$,
 $g_5(x_1,x_2)=-2 x_1+4 x_2+3$. We use the abbreviation $m_{ij\dots}$ for $\max(g_i,g_j,\dots)$ and for those maxima that get expanded, we also display the vector~$\alpha$ that determines the sets $S$ and~$T$. We show three different executions of the algorithm, which however all yield the same final result (combining equal terms is not shown explicitly here).}
\label{fig:AlgComp}
\end{figure}

\begin{exa}
Let $g_1(x_1,x_2)=c_1x_1+c_2x_2$ for some constants $c_1,c_2$ (not both zero) and let $g_i=g_1+i$ for $i=2,3,4$. Then clearly $\max(g_1,g_2,g_3,g_4)=g_4$. However, this simplest-possible answer cannot be found by our algorithm. If we choose $\alpha=(-1,0,0,1)$, then we get $S=\{g_1,g_2,g_3\}$, $T=\{g_4\}$ and hence
$\max(g_1,g_2,g_3,g_4)=\max(g_1,g_2,g_4)+\max(g_1,g_3,g_4)+\max(g_2,g_3,g_4)-\max(g_1,g_4)-\max(g_2,g_4)-\max(g_3,g_4)+g_4$. In contrast, if we choose $\alpha=(-3,1,1,1)$, then $S=\{g_1\}$, $T=\{g_2,g_3,g_4\}$ and we obtain $\max(g_1,g_2,g_3,g_4)=\max(g_2,g_3,g_4)$ as the final result. However, this weakness can easily be cured by ignoring the condition $|R|\geq n+2$ in Lemma~\ref{lem:ST}: now the existence of a vector~$\alpha$ is not guaranteed any more, but if it exists, we perform the corresponding decomposition, otherwise we leave that term unchanged.
\end{exa}

\section{Reducing the Height}
\label{sec:exmp}

In this section, we give an algorithm for writing the maximum of any number of affine-linear functions from $\R^n\to\R$ as a linear combination of maxima of at most $n+1$ of these functions.

Algorithm~\ref{alg:lemma_id} implements the method \textit{reduceMax}  that takes the size of the input space $n$ and a maximum function $f(x)\coloneqq \max(h_{1}(x), \ldots, h_{k}(x))$ with more than $n+1$ constituents, i.e. $R\coloneqq \{h_{1}, \ldots, h_{k}\}$ and $|R| > n+1$. It returns a linear combination of maxima: 
\[g \coloneqq\sum_{i=1}^{l} c_i\cdot\max(h^{i}_{1}, \ldots, h^{i}_{k_{i}}),\] 
where $\{h_{1}^{i}, \ldots, h_{k_{i}}^{i}\} \subset R$ and $k_{i} \leq n+1$ for all $i \in \{1, \ldots, l\}$. The combination is constructed in the following way. 

Firstly, one extracts the linear constituents $h_{1}, \ldots, h_{k}$ from the input maximum $f$ and forms the set $R$. Then the set $R$ is split into two disjoint subsets $S, T \subset R$ such that $R = S\uplus T$. The \textit{split} operation is the implementation of Lemma~\ref{lem:ST} and is described in Algorithm~\ref{alg:lemma_st}. More details about the \textit{split} operation can be found in the second part of the chapter. After splitting the set $R$ into the pair of subsets $S, T$, the linear combination of maxima $g$ is generated. The linear combination of maxima $g$ has the following form:
\begin{equation*}
g(x) = (-1 )^{|S|+1}\sum_{M \subset S} (-1)^{|M|}\max_{h \in M \cup T}h(x), 
\end{equation*}
where the sum runs over all proper subsets of $S$ including the empty set. By Lemma~\ref{lem:id}, $g(x)$ is equal to the input function $f(x)$, and every maximum in $g(x)$ contains at most $|R| -1$ constituents, where the set $R$ contains all constituents of the input function $f(x)$. If any summand of $g(x)$ contains affinely dependent subset of constituents, one replaces it with the corresponding linear combination of maximum by applying the \textit{reduceMax} method recursively. One repeats this simplification procedure until all the summands contain at most $n+1$ constituents. 

The termination of Algorithm~\ref{alg:lemma_id} follows from the fact that after every call on the input maximum with $k$ constituents, \textit{reduceMax} returns a finite combination of maxima where each maximum contains at most $k-1$ constituents.

Algorithm~\ref{alg:lemma_id} uses the method \textit{split} for dividing a set of constituents $R$ into two disjoint sets $S, T$. The method \textit{split} is described in Algorithm~\ref{alg:lemma_st} and it is an implementation of Lemma~\ref{lem:ST}. The algorithm splits the input set of constituents $R$ based on the sign of the vector $\alpha \in \R^{|R|}$ that is a solution of the system of equalities:
\begin{equation*}
    \sum_{i= 1}^{|R|}\alpha_{i}h'_{i} = 0, \quad\sum_{i = 1}^{|R|}\alpha_{i} = 0,
\end{equation*}
where $h'_{i} \coloneqq
\begin{pmatrix}
  \dfrac{\partial h_{i}}{\partial x_{1}}, \ldots,   \dfrac{\partial h_{i}}{\partial x_{n}}
\end{pmatrix}^{T}$ and $h_{i} \in R$ for all $i \in \{1, \ldots, |R|\}$. Due to the fact, that all constituents are linear, it implies that the system of equations is linear:
\begin{equation*}
    W\alpha = 0,
\end{equation*}
where $\alpha \in \R^{|R|}$ and $W\in \R^{(n+1) \times |R|}$ such that:
\begin{equation*}
  W \coloneqq   
  \begin{bmatrix}
       h'_{11}& h'_{21}& \ldots& h'_{|R|1}\\
        \vdots&  &  &\vdots\\
        h'_{1n} & h'_{2n}&\ldots& h'_{|R|n}\\
        1& 1& \ldots& 1
   \end{bmatrix}. 
\end{equation*}
    Solving the given system of linear equations is equivalent to finding the null space $\ker(W)$. If the null space is trivial, one does not need to divide the input set $R$ because Lemma~\ref{lem:ST} is not applicable. Otherwise,  the vector $\alpha$ can be picked as any vector from $\ker(W)$. By iterating  through the entries of the vector $\alpha$, depending on the sign of the entry $\alpha_{i}$, the corresponding constituent $h_{i}$ is assigned either to $T$ or $S$. Note that the condition $k\geq n+2$ in Lemma~\ref{lem:ST} ensures the existence of a non-trivial null space $\ker(W)$.


\begin{exa}
\label{ex:alg}
    Let us take the following function: 
    \begin{align*}
        g(x_{1}, x_{2}) \coloneqq \max(&x_1, x_1 + x_2, x_2 + \max(x_1 + x_2 - 7, x_1 + 6x_2 + 4) \\ &+ 3\min(4x_2, x_1 - 9, x_1 - x_2)).
    \end{align*}
    After applying the function expansion by the rules explained in~\cite{Wang_Sun:05} on~$g$, one receives a linear combination of maxima with 49 summands, where 40 of them contain 5 or more constituents. After applying  Algorithm~\ref{alg:lemma_id} on the expanded version of~$g$, the given sum transforms into a new one with only 7 summands:
    \begin{align*}     
   \hat{g}(x_1, x_2) \coloneqq  &\max(x_1, x_1 + 19x_2 + 4) \\& - \max(x_1 + x_2, 4x_1 + 4x_2 + 4, 4x_1 + 7x_2 - 23)  \\
&- \max(x_1 + 19x_2 + 4, 4x_1 + 4x_2 + 4)\\& + \max(x_1, x_1 + 
x_2, 4x_1 + 7x_2 - 23) \\ &- \max(x_1, x_1 + 19x_2 + 4, 4x_1 + 
7x_2 - 23) \\&+ \max(x_1 + x_2, 4x_1 + 4x_2 + 4)\\ &+ \max(x_1 + 19x_2 + 4, 4x_1 + 4x_2 + 4, 4x_1 + 7x_2 - 23).
 \end{align*}
 The  function $\hat{g}$ contains maxima with at most 3 constituents, in accordance with Theorem~\ref{thm:ws}. However, an example can be found where the number of maxima increases after applying Algorithm~\ref{alg:lemma_id}, compared to the output after applying the rules from \cite{Wang_Sun:05}. For instance, it holds for the function:
 \begin{align*}
     f(x_1, x_2) \coloneqq \max(&6 x_1 + 5 x_2 - 3, 8x_2 - 2, -3 x_1 -5x_2 - 4, \\
 &\max(12x_1 - 4x_2 + 1, -7x_1 + 8x_2 + 12) + 3x_1 - 10,\\
 &\min(-3x_1 + 4 x_2 - 5, 8x_1 + 2)).
 \end{align*}
 After repeating two transformations, we receive two linear combinations with the number of maxima 3 and 5, respectively, with the final form:
 \begin{align*}
    \hat{f}(x_1, x_2)\coloneqq & -\max(-4x_1 + 8x_2 + 2, 15x_1 - 4x_2 - 9)  \\ &+\max(8x_2 - 2, -4x_1 + 8x_2 + 2, 6x_1 + 5x_2 - 3) \\
    &+\max(-4x_1 + 8x_2 + 2, -3x_1 - 5x_2 - 4, 15x_1 - 4x_2 - 9) \\
    &+ \max(-4x_1 + 8x_2 + 2, 6x_1+ 5x_2 - 3, 15x_1 - 4x_2 - 9) \\
    &- \max(-4x_1 + 8x_2 + 2, 6x_1 + 5x_2 - 3).
 \end{align*}

Although Algorithm~\ref{alg:lemma_id} can either reduce the number of summands in the final linear combination of maxima or increase it, the final combination seems to be invariant, not depending on the outcome of Algorithm~\ref{alg:lemma_st}. More precisely, we conjecture that the output of Algorithm~\ref{alg:lemma_id} is independent of how  the set of constituents is split by Algorithm~\ref{alg:lemma_st}, as illustrated in Figure~\ref{fig:AlgComp}.
\end{exa}

\begin{algorithm}[!htpp]
    \caption{split (Lemma~\ref{lem:ST})}
    \label{alg:lemma_st}
    \KwInput{$\{h_{1}, \ldots, h_{k}\}$ -- set of linear constituents, 
    where $h_{i}\colon\R^{n} \rightarrow \R, \forall i \in \{1, \ldots, k\}$}
    \KwOutput{$(S,T)$ -- disjoint decomposition of $R$}
    $W \coloneqq 
    \begin{pmatrix}
        \begin{pmatrix}
            h'_{1}\\ 1 
    \end{pmatrix}, \ldots, 
    \begin{pmatrix}
            h'_{k}\\ 1
    \end{pmatrix}
    \end{pmatrix} $\\
    $K \coloneqq \ker(W)$\\
    \If {$K$ is zero}{
        \KwReturn{$\emptyset, \{h_{1}, \ldots, h_{k}\}$}
    }
    
    $\alpha \coloneqq \text{pick a non-zero vector from }K$\\
    $v \coloneqq \text{constant terms of } h_{1}, \ldots, h_{k}$\\
    $S \coloneqq \emptyset$\\
    $T \coloneqq \emptyset$ \\
    $c \coloneqq v^{T}\alpha$\\
    \For{$i \in \{1, \ldots, k\}$}{
    \If{$\alpha_{i} > 0$}{
    $T\coloneqq T \cup \{h_{i}\}$
    }\Else{
    $S \coloneqq S \cup \{h_{i}\}$
    }
    }
    \If {$c > 0$}{
    \KwReturn{$S,T$}
    }
    \KwReturn{$T,S$}
\end{algorithm}

\begin{algorithm}[!htpp]
    \caption{reduceMax (Lemma~\ref{lem:id})}
    \label{alg:lemma_id}
    \KwInput{$g \coloneqq \max(h_{1}, \ldots, h_{k})$ -- max function, where $g \colon\R^{n}\rightarrow \R$.}
    \KwOutput{Linear combination of max functions.}
    $R \coloneqq \{h_{1}, \ldots, h_{k}\}$ \\
      $(S, T)\coloneqq \text{split} (R)$ \\
      $A \coloneqq \max(T)$ \\
      \If {$S=\emptyset$}{   
      \KwReturn{$A$}
      }
      \If{$|T|>1$}{
      $A \coloneqq \text{reduceMax}(A)$
      }   
      \For{all proper subsets $M$ of $S$}{
      $P\coloneqq M \cup T$\\
        $A \coloneqq A + (-1)^{|M|}\text{reduceMax}(\max(P))$
      }
    \KwReturn{$(-1)^{|S|+1}A$}

\end{algorithm}


\section{Duality and Convex Polyhedra}
\label{sec:duality}

Let $n$ be a positive integer. Instead of all piecewise linear functions on $\R^n$, we consider in this section the subset
of convex and positively homogeneous piecewise linear functions of degree~1, i.e., all functions $f\colon\R^n\to\R$ that are convex and
satisfy $f(\lambda x)=\lambda f(x)$ for all $x\in\R^n$ and $\lambda\ge 0$. We denote this subset by ${\cal F}_n$.
There is a useful bijective correspondence $\tau\colon{\cal F}_n\to {\cal P}_n$, where we define ${\cal P}_n$ as the set of all
non-empty convex polytopes in~$\R^n$. For each $f\in{\cal F}_n$, we define $\tau(f)$ as the subset of all vectors $h\in\R^n$ such that
$\langle h,x\rangle\le f(x)$ for all $x\in\R^n$. Conversely, if $P\in{\cal P}_n$, then $\tau^{-1}(P)$ is the support function
$x\mapsto \sup\{\langle h,x\rangle\mid h\in P\}$. By compactness, the supremum is a maximum.

For two polytopes $P,Q\in{\cal P}_n$, the Minkowski sum $P+Q$ is defined as the convex polytope $\{a+b\mid a\in P, b\in Q\}$ as illustrated in Figure~\ref{fig:tau_rep}.
Let $d\in\R^n\setminus\{0\}$ be a vector ($d$ stands for ``direction''). For any polytope $P\in{\cal P}_n$, we define the face
\[
  S_d(P) = \{ x\in P\mid \langle x,d\rangle = \max_{y\in P}\,\langle y,d\rangle \} .
\]

\begin{prop} \label{prop:madd}
Let $P,Q\in{\cal P}_n$ be polytopes. Let $d\in\R^n\setminus\{0\}$ be a direction vector. Then
\begin{enumerate}
\renewcommand{\labelenumi}{\alph{enumi})}
\item The map $\tau$ is an isomorphism of semigroups:
  \[
  \tau^{-1}(P)+\tau^{-1}(Q)=\tau^{-1}(P+Q).
  \]
\item Taking faces is an endomorphism of semigroups:
  \[
  S_d(P+Q)=S_d(P)+S_d(Q).
  \]
\end{enumerate}
\end{prop}

\begin{proof}
See \cite[Lemma 2.1.4]{Gritzmann_Sturmfels:93}.
\end{proof}

\begin{prop} \label{prop:cancel}
Minkowski addition is cancellable: if $A+C=B+C$, then $A=B$.
\end{prop}

\begin{proof}
This is well-known, and we can prove it easily by translation to functions. Assume $A+C=B+C$.  Then
\[ \tau^{-1}(A)+\tau^{-1}(C)=\tau^{-1}(A+C)=\tau^{-1}(B+C)=\tau^{-1}(B)+\tau^{-1}(C), \]
hence $\tau^{-1}(A)=\tau^{-1}(B)$ and therefore $A=B$.
\end{proof}

The face $S_d(P)$ is contained in the hyperplane $H_{d,c}=\{x \mid \langle x,d\rangle = c\}$, where $c:=\max_{y\in P}\,\langle y,d\rangle$.
For the induction proof in the next section, we need to identify $H_{d,c}$ with $\R^{n-1}$. To this end, we translate the hyperplane
to $H_{d,0}$ by a translation vector $cv$ that satisfies $\langle d,v\rangle=-1$, and then we apply an isomorphism $\phi\colon H_{d,0}\to\R^{n-1}$.
The face $\phi(S_d(P)+cv)$ is denoted by $F_d(P)$. It depends not only on $d$, but also on the choice of $v$ and $\phi$; 
but we may choose $v_d$ and $\phi_d$ for every $d$ once and for all, subject to the condition $v_{-d}=-v_d$ and $\phi_d=\phi_{-d}$.

\begin{figure}
\begin{center}
\setlength{\tabcolsep}{20pt}
\begin{tabular}{@{}ccc@{}}
\includegraphics[width=0.25\textwidth]{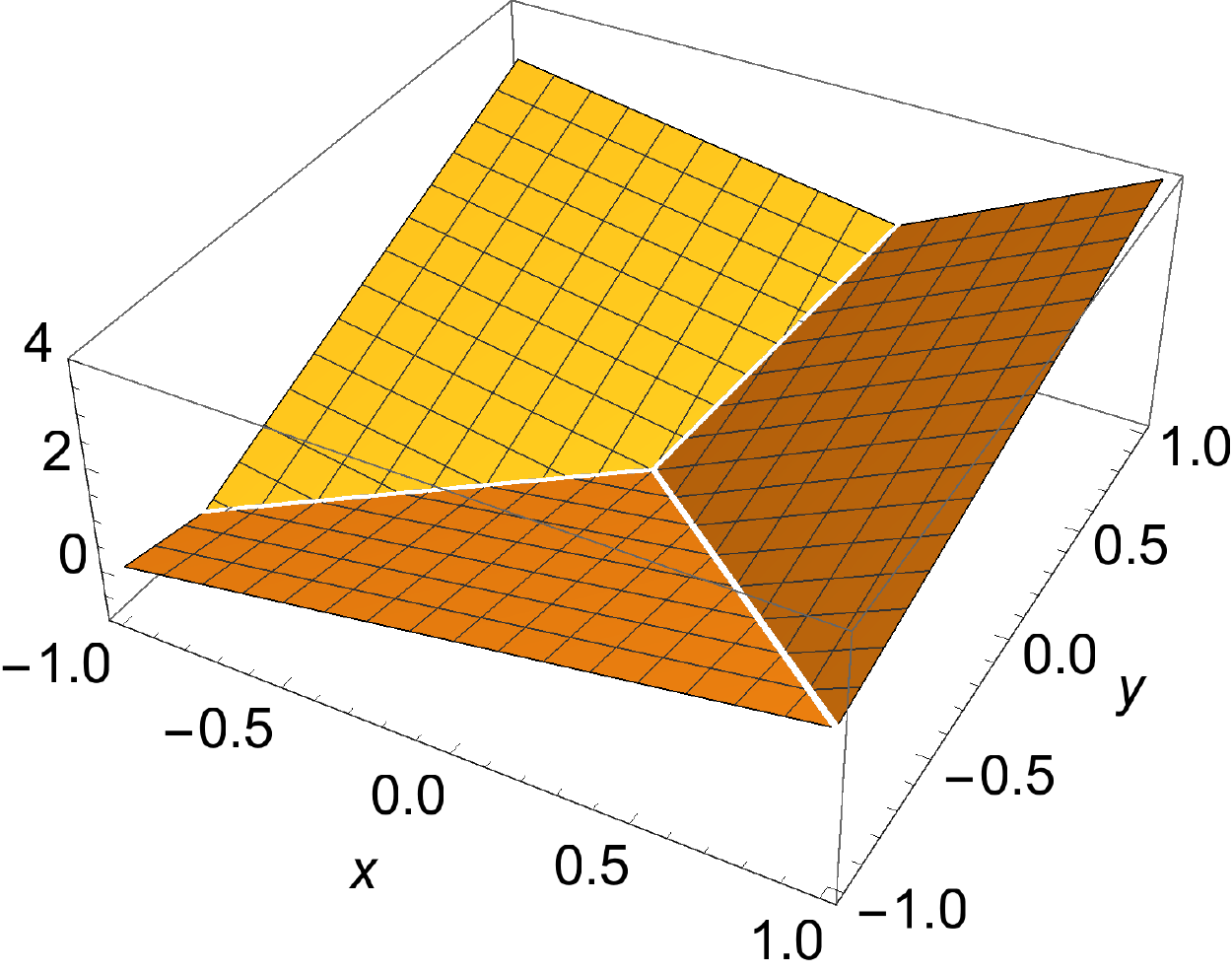} &
\includegraphics[width=0.25\textwidth]{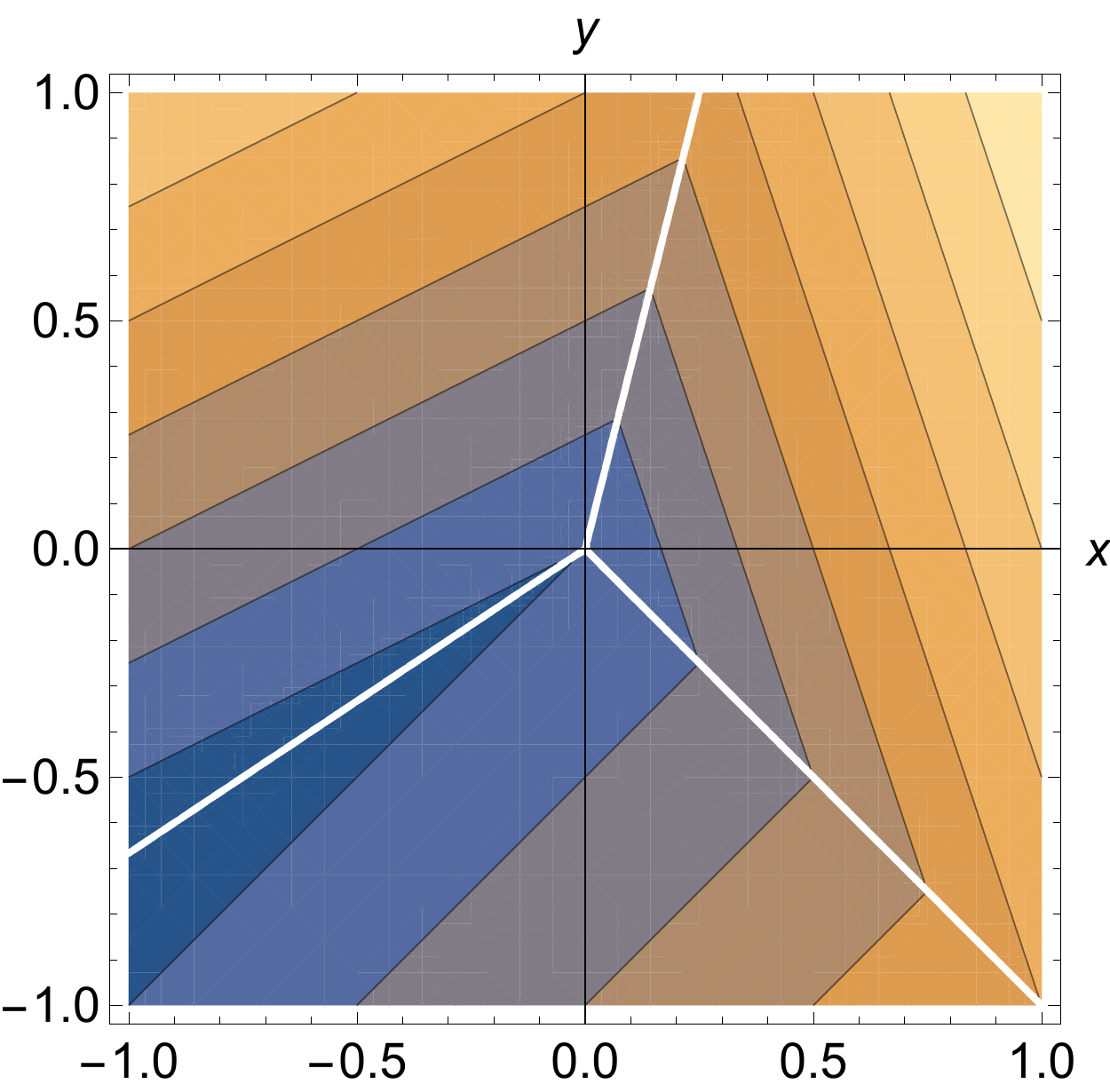} &
\includegraphics[width=0.25\textwidth]{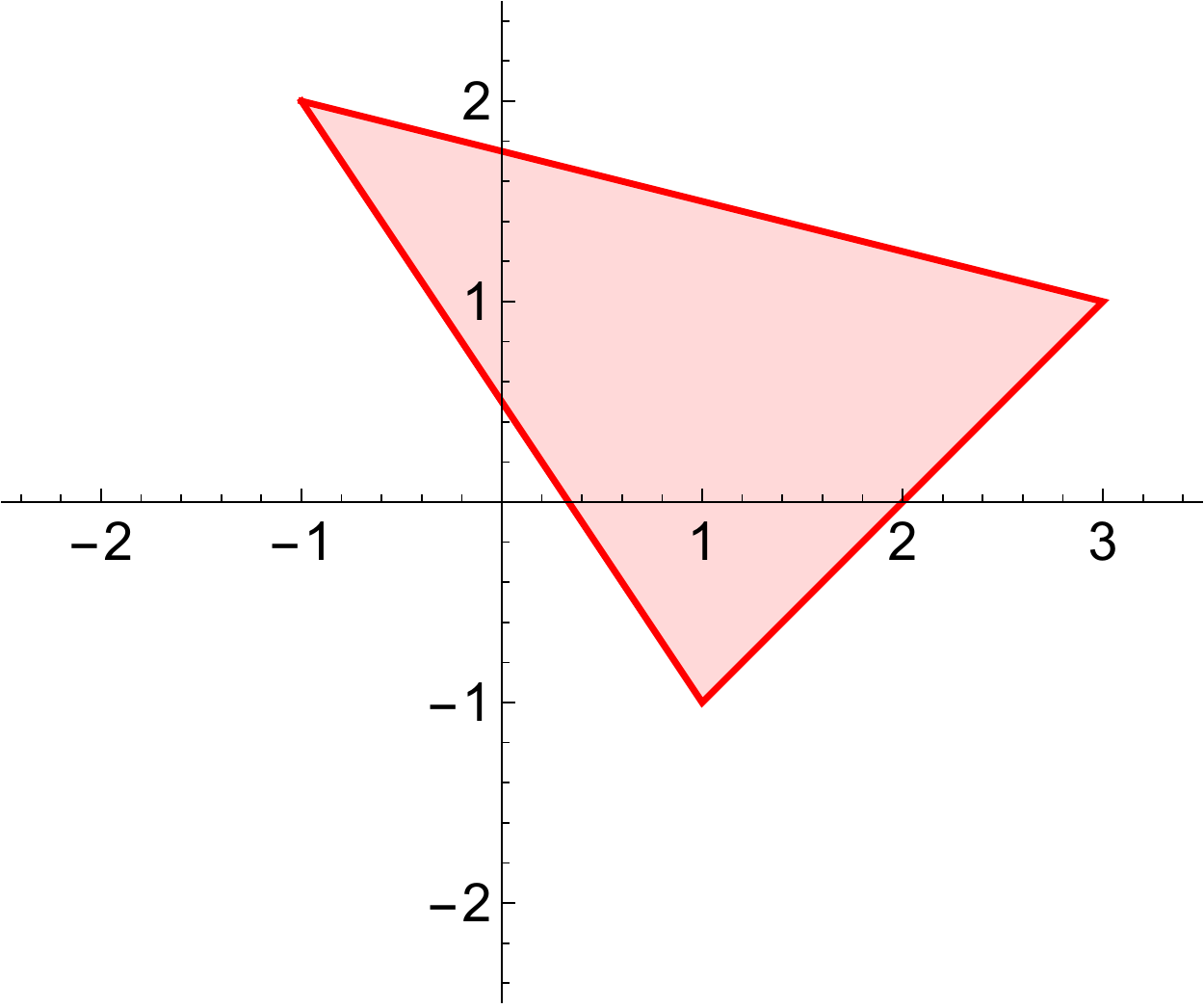} \\[12pt]
\includegraphics[width=0.25\textwidth]{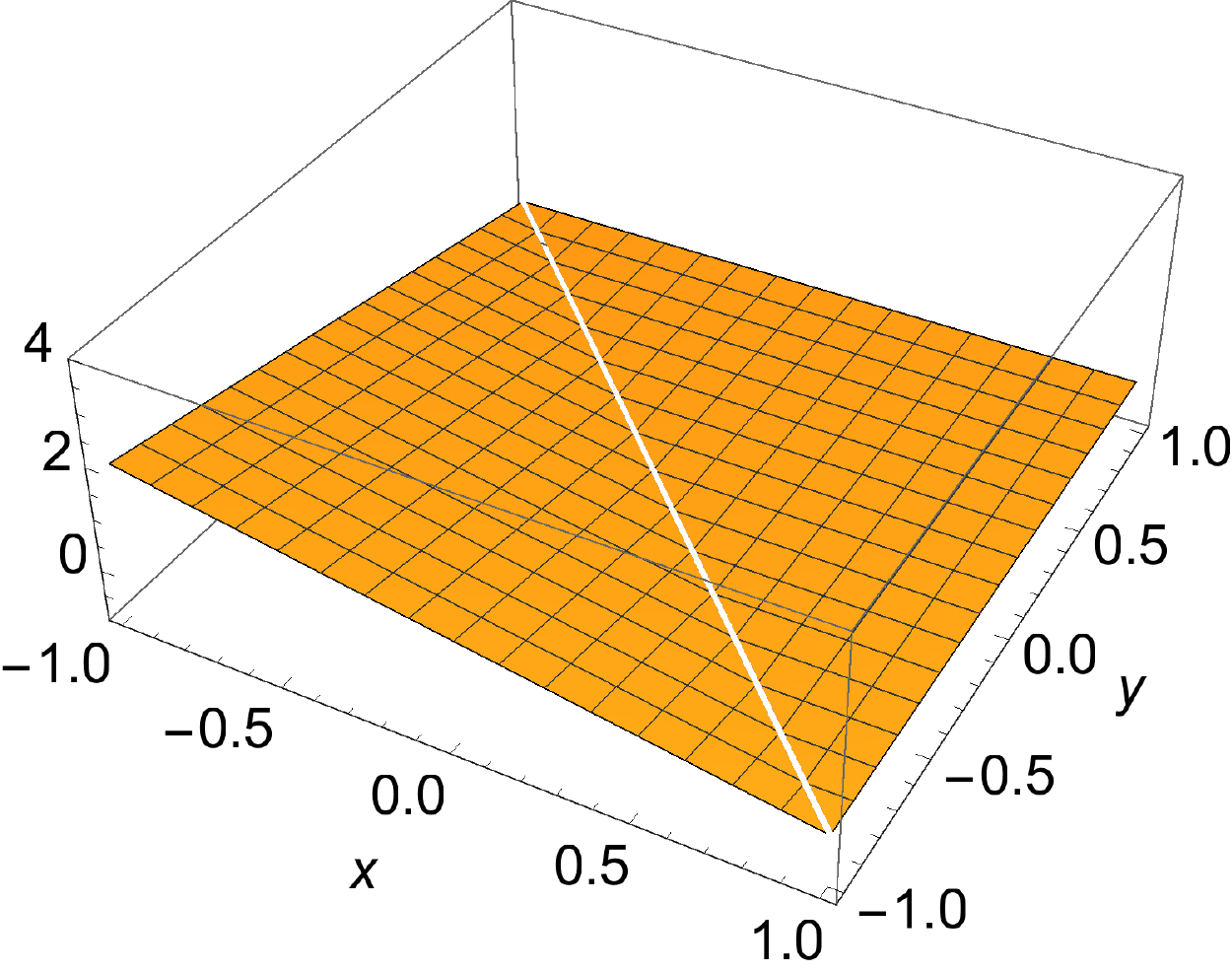} &
\includegraphics[width=0.25\textwidth]{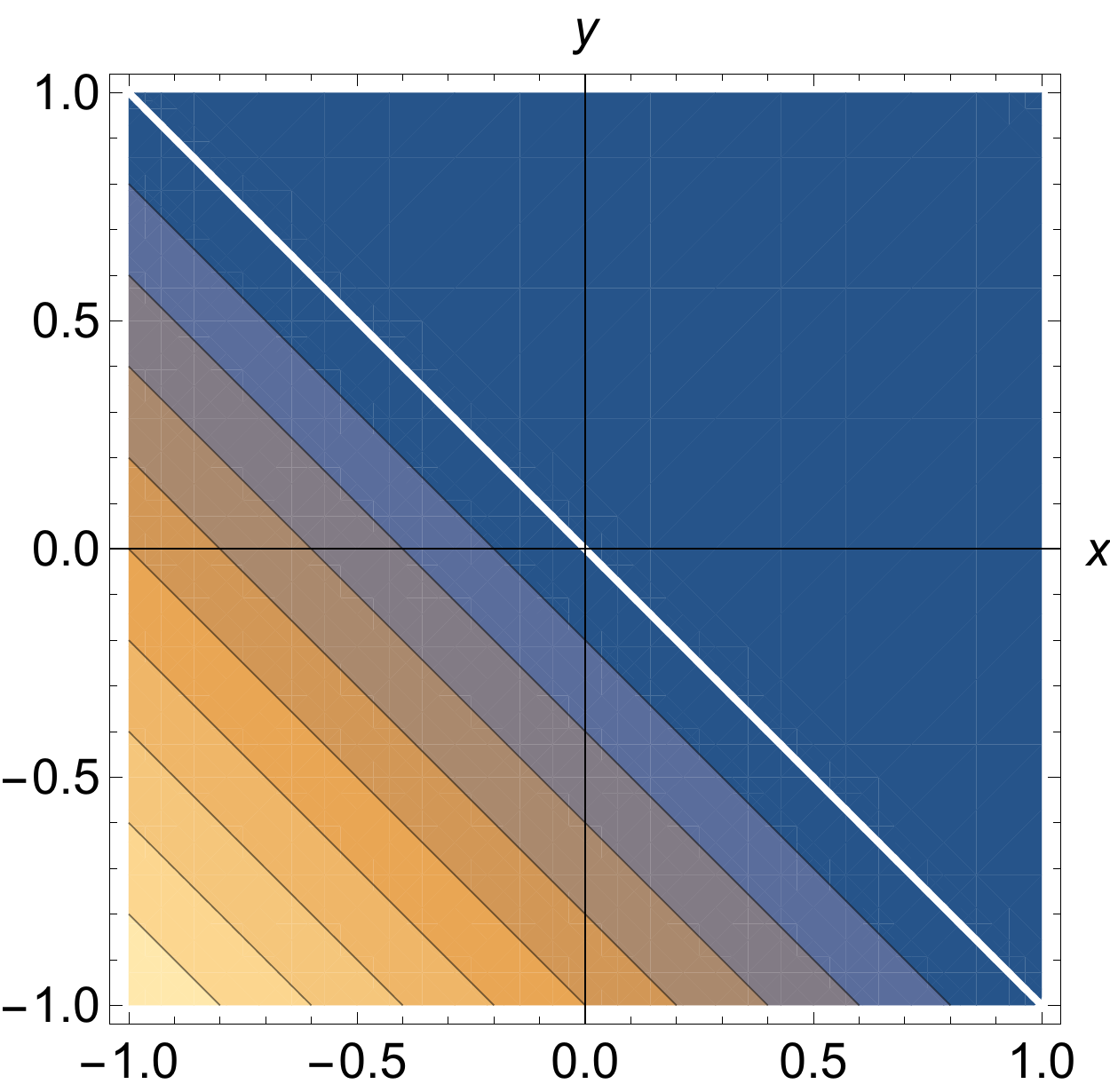} &
\includegraphics[width=0.25\textwidth]{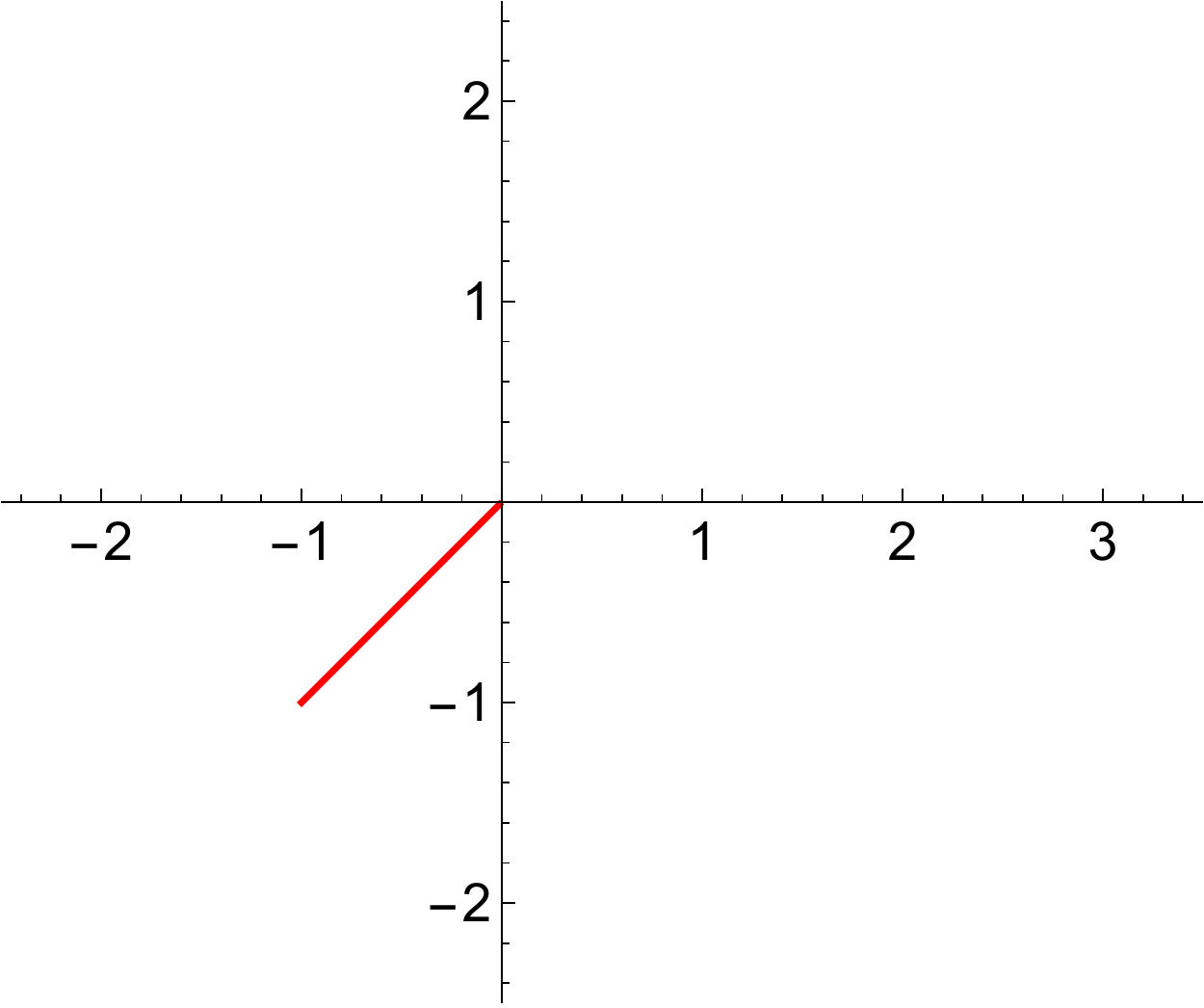} \\[12pt]
\includegraphics[width=0.25\textwidth]{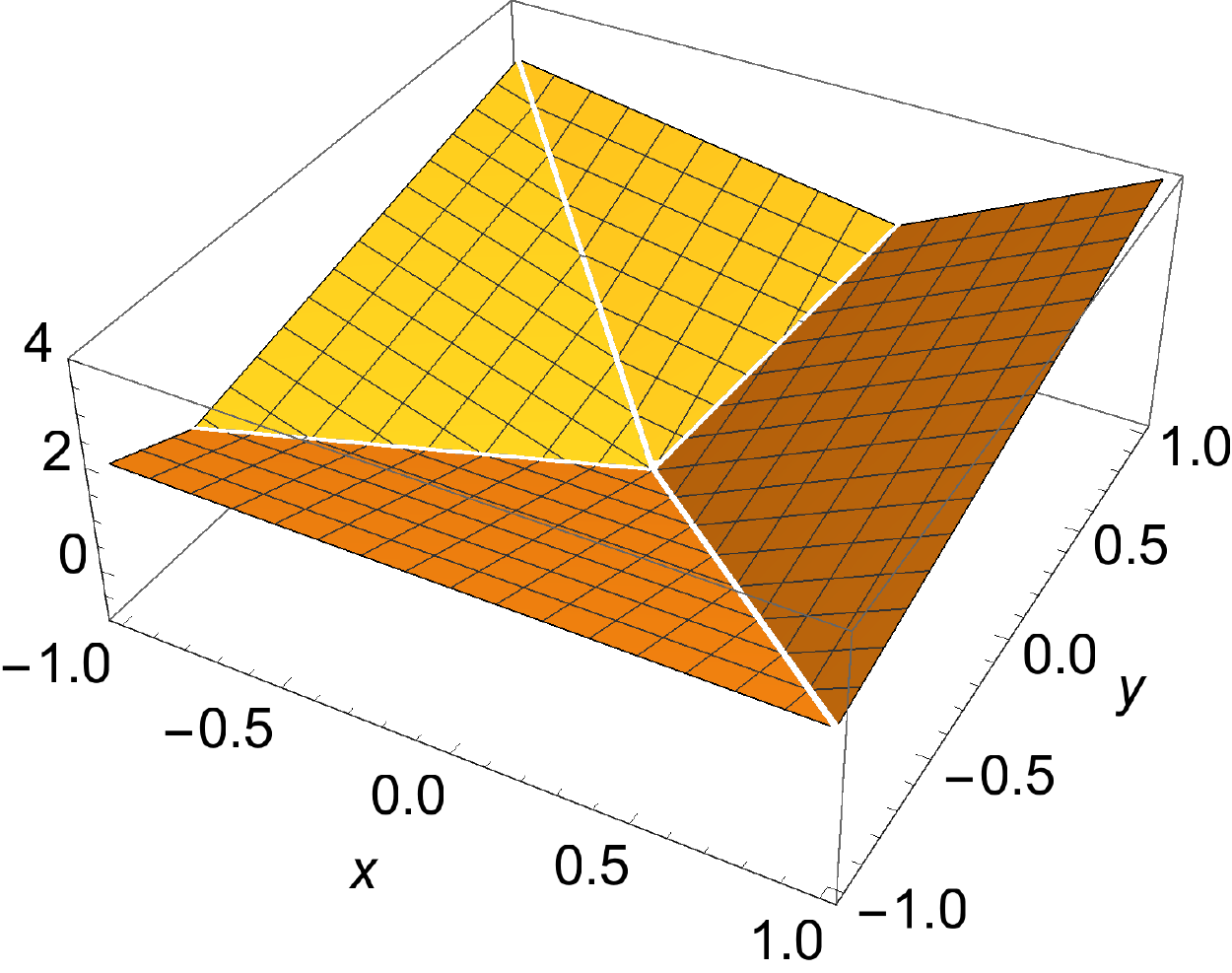} &
\includegraphics[width=0.25\textwidth]{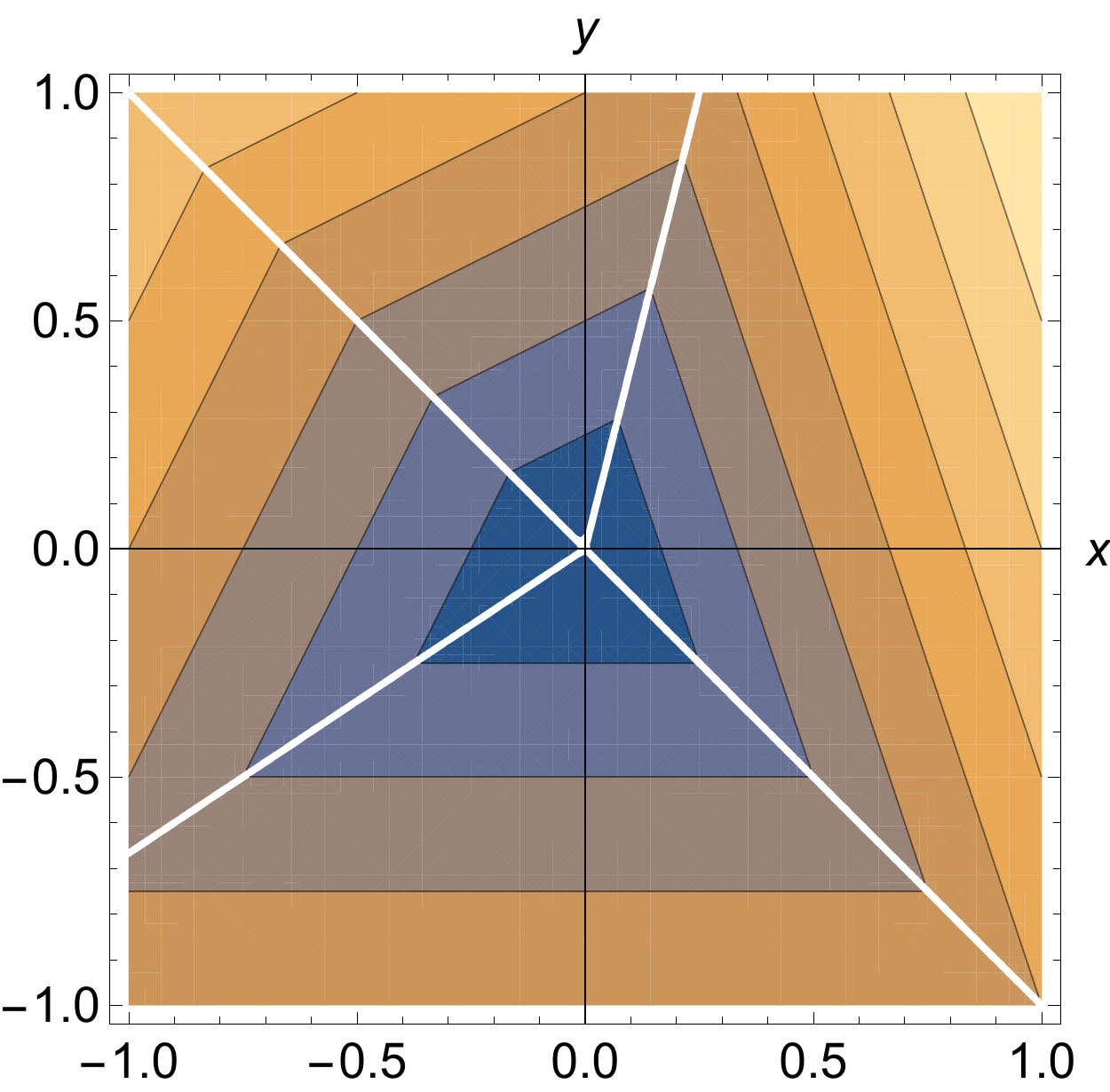} &
\includegraphics[width=0.25\textwidth]{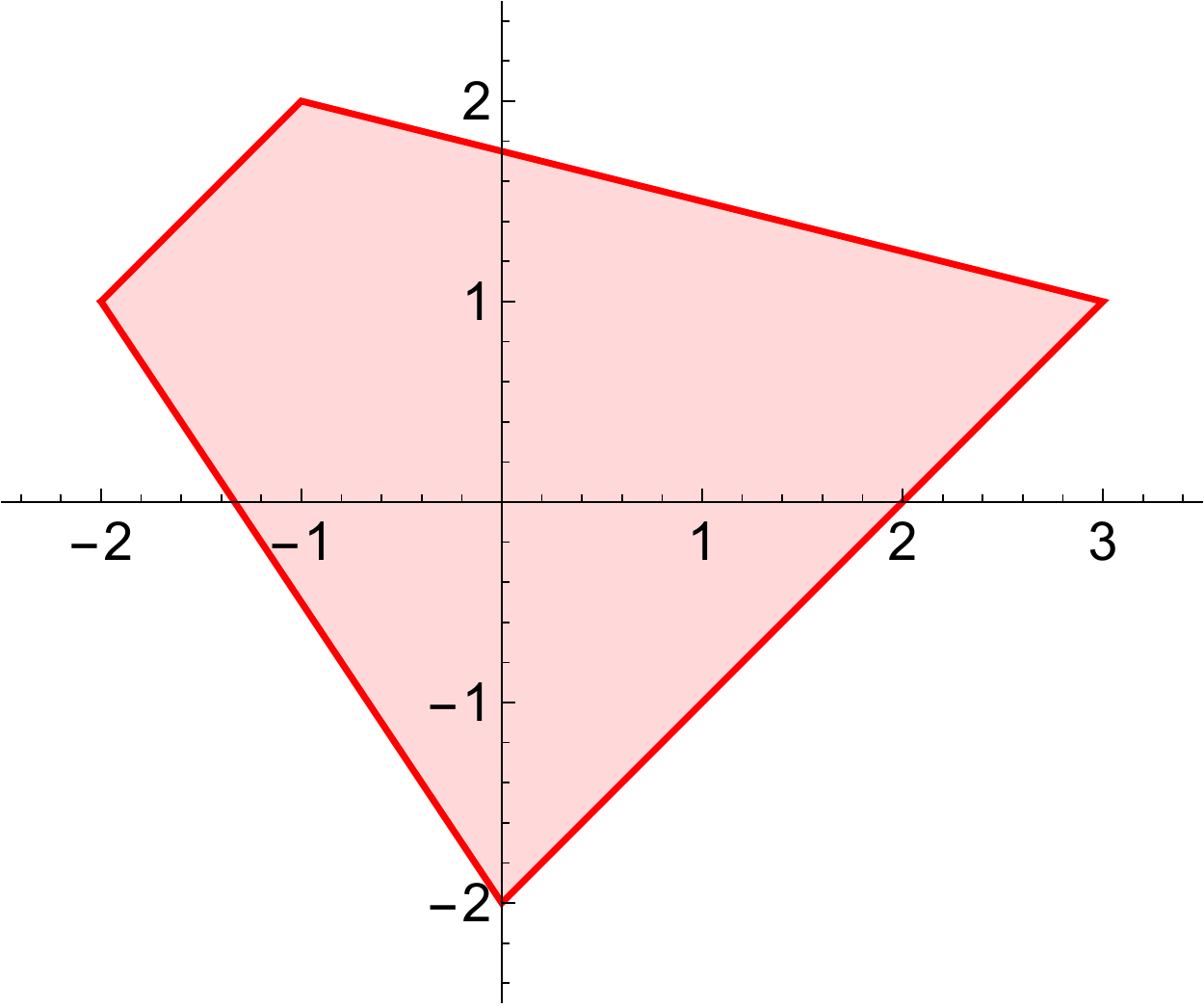}
\end{tabular}
\caption{Two positively homogeneous piecewise linear functions, $f(x,y)=\max(x-y, 3x+y, -x+2y)$ (first row) and $g(x,y)=\max(0,-x-y)$
(second row), as well as their sum $h=f+g$ (third row). Each row shows the function $\R^2\to\R$ as a three-dimensional plot (first column) and as a contour plot (second column). The third column shows the corresponding polygon $\tau(f)$ etc. in~${\cal P}_n$. One clearly sees how $h(x,y)=\max(-2y, -2x+y, 3x+y, -x+2y)$ produces the Minkowski sum $\tau(h)=\tau(f)+\tau(g)$ of polygons.}
\label{fig:tau_rep}
\end{center}
\end{figure}

\section{A Lower Bound for the Height}
\label{sec:lb}

Recall that an $n$-simplex is a polytope which is the convex hull of $n+1$ points that are affinely independent. The faces of simplices
are again simplices.

A polytope $P\in{\cal P}_n$ is called a {\em zero volume} polytope if and only if it has no interior points; this is the case if and only if
it is contained in a hyperplane. If $P$ has zero volume and $d\in\R^n\setminus\{0\}$ is a direction vector, then one of the following two cases holds:
\begin{itemize}
\item either $S_d(P)=S_{-d}(P)=P$, 
\item or both faces $F_d(P)$ and $F_{-d}(P)$ have zero volume as polytopes in ${\cal P}_{n-1}$.
\end{itemize}

We say that a polytope $P\in{\cal P}_n$ is a {\em zero-summand} if and only if there are convex polytopes $P_1,\dots,P_r,Q_1,\dots,Q_s\in{\cal P}_n$ 
of zero volume such that $P+P_1+\dots+P_r=Q_1+\dots+Q_s$, as illustrated in Figure~\ref{fig:zerosmnd}.

\begin{figure}
\begin{center}
\begin{tikzpicture}[scale=0.5]
    \draw[-,thick] (-13,1.5) -- (-11,0.5);
    \node (p1) at (-10,1) {$+$};
    \draw[-,thick] (-9,0.5) -- (-9,1.5);
    \node (p2) at (-8,1) {$+$};
    \draw[-,thick] (-7,0) -- (-4,2);
    \node (eq) at (-2,1) {$=$};
    \draw[-,thick] (0,0) -- (0,1) -- (3,3) -- (5,2) -- (5,1) -- (2,-1) -- (0,0);
    \node (p) at (2.5,1) {$P$};
\end{tikzpicture}
\caption{The polytope~$P$ is a zero-summand, since it can be written as a Minkowski sum of three line segments.}
\label{fig:zerosmnd}
\end{center}
\end{figure}
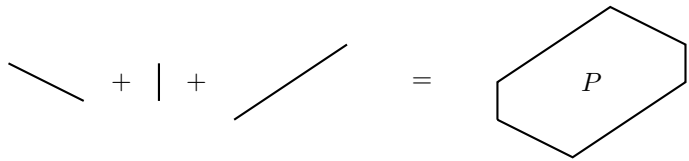

\begin{lem} \label{lem:goingdown}
Let $P\in{\cal P}_n$ be a zero-summand. Let $d\in\R^n\setminus\{0\}$ be a direction vector such that $F_{-d}(P)\in{\cal P}_{n-1}$ has zero volume.
Then $F_{d}(P)$ is also a zero-summand.
\end{lem}

\begin{proof}
Since $P\in{\cal P}_n$ is a zero summand, there exist convex polytopes $P_1,\dots,P_r$, $Q_1,\dots,Q_s\in{\cal P}_n$ of zero volume such that 
\[ P+P_1+\dots+P_r=Q_1+\dots+Q_s. \] 
We may assume, without loss of generality, that the zero volume polytopes $P_i$ or $Q_j$ that have a face $F_d(P_i)=F_{-d}(P_i)$ of nonzero volume 
-- or similarly with $Q_j$ -- are $P_1,\dots,P_k$ and $Q_1,\dots,Q_l$, for some $k\le r$ and $l\le s$.
Set $A:=P+P_1+\ldots+P_r$; therefore also $A=Q_1+\dots+Q_s$. We apply $F_d$ to $A$ and neglect all polytopes of zero volume.
This yields 
\[ F_d(A)=F_d(P)+F_d(P_1)+\dots+F_d(P_k)\] plus polytopes of zero volume, and
\[ F_d(A)=F_d(Q_1)+\ldots+F_d(Q_l) \] plus polytopes of zero volume.
Now we apply $F_{-d}$ to $A$ and neglect all polytopes of zero volume, taking into account the equations $F_d(P_i)=F_{-d}(P_i)$ for $i=1,\dots,k$ 
and $F_d(Q_j)=F_{-d}(Q_j)$ for $j=1,\dots,l$, and the assumption that $F_{-d}(P)\in{\cal P}_{n-1}$ has zero volume.
This yields 
\[ F_{-d}(A)=F_d(P_1)+\dots+F_d(P_k) \]
plus polytopes of zero volume, and 
\[ F_{-d}(A)=F_d(Q_1)+\dots+F_d(Q_l) \]
plus polytopes of zero volume.
Summing up, we get that $F_d(A)+F_{-d}(A)$ is equal to both sides 
of the equation 
\[F_d(P)+F_d(P_1)+\dots+F_d(P_k)+F_d(Q_1)+\dots+F_d(Q_l) \]
\[ = F_d(P_1)+\dots+F_d(P_k)+F_d(Q_1)+\dots+F_d(Q_l) \]
modulo polytopes of zero volume.
By Proposition~\ref{prop:cancel}, it follows that $F_d(P)$ is a zero summand.
\end{proof}

\begin{cor} \label{cor:n0summand}
An $n$-simplex in ${\cal P}_n$ is not a zero summand.
\end{cor}

\begin{proof}
Induction on $n$: if $n=1$, then the zero volume polytopes are single points,
and therefore the zero summands are also single points. 
But a 1-simplex is a line segment of positive length and not a single point.

If $\Delta\in{\cal P}_n$ is an $n$-simplex for some $n>1$, then there is a direction vector $d$ such that $F_{-d}(\Delta)$ is a point and $F_d(\Delta)$
is an $(n-1)$-simplex in ${\cal P}_{n-1}$. By the induction hypothesis, $F_d(\Delta)$ is not a zero summand. Also, $F_{-d}(\Delta)$ has zero volume. By Lemma~\ref{lem:goingdown}, applied in
contraposition, it follows that $\Delta$ is not a zero summand.
\end{proof}

\begin{lem} \label{lem:almost}
Let $g_0,\dots,g_n\colon\R^n\to\R$ be linear functions whose derivatives are affinely independent. Then the positively homogeneous function
$\max(g_0,\dots,g_n)$ is not a linear combination of maxima of at most $n$ linear functions.
\end{lem}

\begin{proof} 
Assume, indirectly, that $\max(g_0,\dots,g_n)=\alpha_1f_1+\dots+\alpha_kf_k$, where $f_1,\dots,f_k$ are maxima of at most $n$ linear functions.
Without loss of generality, we may assume that the $\alpha_i$ are either $1$ or $-1$.
Let us assume that $\alpha_1=\dots=\alpha_l=1$ and $\alpha_{l+1}=\dots=\alpha_k=-1$.
Hence $\max(g_0,\dots,g_n)+f_{l+1}+\dots+f_k=f_1+\dots+f_l$. 
By Proposition~\ref{prop:madd}, we obtain
\[
  \tau(\max(g_0,\dots,g_n)) + P_{l+1}+\dots +P_k = P_1+\dots+P_l .
\]
For $i=1,\dots,k$, the function $f_i$ is in ${\cal F}_n$, and $P_i:=\tau(f_i)$ is a
zero volume polytope because it is the convex hull of at most $n$ points. 
This shows that $\tau(\max(g_0,\dots,g_n))$ is a zero summand. But this contradicts Corollary~\ref{cor:n0summand}, because $\tau(\max(g_0,\dots,g_n))$ is an $n$-simplex.
\end{proof}

\begin{thm}
The function $f:=\max(0,x_1,\dots,x_n)\colon\R^n\to\R$ is not a linear combination of
maxima of less than $n+1$ affine-linear functions.
\end{thm}

\begin{proof}
By Lemma~\ref{lem:almost}, the function $f$ is not a linear combination of maxima of less than $n+1$ linear functions:
Assume, indirectly, that there are integers $s,k_1,\dots,k_s$ with $k_i\le n$ for $i=1,\dots,s$, real numbers $\alpha_1,\dots,\alpha_s$ and linear functions $g_{i,j}\colon\R^n\to\R$, $i=1,\dots,s$, $j=1,\dots,k_i$ such that 
\[ m_i:=\max_{j=1,\dots,k_i}(g_{i,j})\mbox{ and } f=\sum_{i=1}^s\alpha_i m_i . \]
We will then construct a representation of $f$ as a linear combination of maxima of less than $n+1$ linear functions, giving a contradiction.

For $i=1,\dots,s$, we proceed as follows. We define $c_i:=m_i(\zero)$. 
We may assume without loss of generality that there exists $r_i\le k_i$ such that $g_{i,j}(\zero)=c_i$ if $j\le r_i$ and $g_{i,j}(\zero)$ if $j>r_i$. 
For $j=1,\dots,r_i$, we set $h_{i,j}:=g_{i,j}-c_i$. Then $h_{i,j}(\zero)=0$, which implies that the functions $h_{i,j}$ are all linear. 
Now we set
\[ n_i:=\max_{j=1,\dots,r_i}(h_{i,j})\mbox{ and } e=\sum_{i=1}^s\alpha_i n_i . \]
Then $e$ is a linear combination of maxima of less than $n+1$ linear functions. We will prove that $e=f$, which will finish the indirect proof.

Let $U$ be a small neighborhood of $\zero$ such that for each $i$, we have

\[ \max_{j=1,\dots,k_i}g_{i,j}=\max_{j=1,\dots,r_i}g_{i,j}= \max_{j=1,\dots,r_i}(h_{i,j}+c_i) = \left(\max_{j=1,\dots,r_i}h_{i,j}\right)+c_i \]
inside $U$. Then we have $m_i=n_i-c_i$ inside $U$ and therefore
\[ f = \sum_{i=1,\dots,s}\alpha_i(n_i+c_i)=e+\sum_{i=1}^s c_i.\]
However, $\sum_{i=1}^s c_i=f(\zero)=0$, and therefore $f=e$ inside $U$. Both functions $f$ and $e$ are positively homogeneous; so, if they coincide in $U$, then they coincide everywhere.
\end{proof}

\section{Conclusion} 

It has been shown that any piecewise linear function $f \colon \R^{n} \rightarrow \R$ can be represented as a linear combination of maxima with at most $n+1$ arguments, where the linear arguments of each maximum are picked from the set of affine-linear parts of the function~$f$. We develop an algorithm for calculating this representation. It is an open question, whether the derived representation is invariant under certain choices that can be made inside Algorithm~\ref{alg:lemma_id}. After running a series of experiments, as illustrated in Example~\ref{ex:alg}, we conjecture that this is the case. Proving this conjecture could be a possible direction for future research.

By proving that the function $\max(0, x_{1}, \ldots, x_{n})$ is not a linear combination of maxima of less than $n+1$ affine-linear functions, we confirm the optimal representation conjecture formulated by Shuning Wang and Xusheng Sun in~\cite{Wang_Sun:05}. Using these two contributions, we can state that every piecewise linear function~$f$ can be expressed as a ReLU neural network with at most $\lceil \log_{2}(n+1) \rceil +1$ layers. This network can be derived from the function representation provided by Algorithm~\ref{alg:lemma_id}. Each hidden layer of this neural network contains less than or equal to $2n + N$ neurons, where $n$ is the dimension of the input space, and where $N$ is the number of maxima in the obtained representation.

\bibliographystyle{amsalpha}
\bibliography{main}

\end{document}